\documentclass[amsart]{amsart}

\title{A lower bound for the determinantal complexity of a hypersurface}

\author[J.~Alper]{Jarod Alper}
\author[T.~Bogart]{Tristram Bogart}
\author[M.~Velasco]{Mauricio Velasco}
\thanks{During the preparation of this paper, the first author was partially supported by the Australian Research 
Council grant DE140101519. The second and third authors were partially supported by the FAPA funds from Universidad de los Andes.}

\date{May 8, 2015}

        \usepackage{latexsym}
        \usepackage{amssymb}
        \usepackage{amsmath}
        \usepackage{amsfonts}
        \usepackage{amsthm}
        \usepackage[hypertexnames=false]{hyperref}
        \ifpdf
          \usepackage[all,knot,poly,2cell]{xy}
        \else
          \usepackage[all,knot,poly,2cell,dvips]{xy}
        \fi
        \newdir{(}{{}*!/-5pt/@^{(}}
             \UseAllTwocells
        \usepackage{xspace}
        \usepackage{comment}
        \usepackage{fixltx2e}
        \usepackage[toc,page]{appendix}
        \usepackage{url}
        \usepackage{color}
        \usepackage{fixltx2e}
        \usepackage{enumitem}
        \setlist[enumerate]{font=\normalfont}
        \usepackage{mathtools}

        \usepackage{eucal}

        \theoremstyle{plain}
        \newtheorem{theorem}{Theorem}[section]
        \newtheorem{corollary}[theorem]{Corollary}

        \newtheorem{proposition}[theorem]{Proposition}

        \theoremstyle{definition}
        \newtheorem{definition}[theorem]{Definition}

        \newtheorem*{example*}{Example}

        \theoremstyle{remark}
        \newtheorem{remark}[theorem]{Remark}
        \newtheorem*{remark*}{Remark}  

\numberwithin{equation}{section}

\newcommand{\PP}{\mathbb{P}}

\DeclareMathOperator{\Sing}{Sing}

\newcommand{\Sym}{\operatorname{Sym}}

\newcommand{\perm}{\operatorname{perm}}
\newcommand{\codim}{\operatorname{codim}}
\newcommand{\rank}{\operatorname{rank}}

\newcommand{\co}{\colon}
\newcommand{\rX}{\mathrm{X}}

\renewcommand{\det}{\operatorname{det}}
\newcommand{\V}{{\rm V}}
\newcommand{\bV}{{\mathbf V}}

\DeclareMathOperator{\Jac}{Jac}

\DeclareMathOperator{\im}{im}
\DeclareMathOperator{\dc}{dc}
\DeclareMathOperator{\VP}{VP}
\DeclareMathOperator{\VNP}{VNP}

\begin{document}

\setcounter{tocdepth}{1}
\begin{abstract}
We prove that the determinantal complexity of a hypersurface of degree $d > 2$ is bounded below by one more than the codimension of the singular locus, provided that this codimension  is at least $5$.  As a result, we obtain that the determinantal complexity of the $3 \times 3$ permanent is $7$.  We also prove that for $n> 3$, there is no nonsingular hypersurface in $\PP^n$  of degree $d$ that has an expression as a determinant of a $d \times d$ matrix of linear forms while on the other hand for $n \le 3$, a general determinantal expression is nonsingular.  Finally, we answer a question of Ressayre by showing that the determinantal complexity of the unique (singular) cubic surface containing a single line is $5$.
\end{abstract}
\maketitle

\section{Introduction}

Let $k$ be a field. For a positive integer $m$, let $X=(x_{ij})_{1\leq i,j\leq m}$ denote an $m \times m$ matrix of linear forms.  Let $\det_m=\det(X)\in k[x_{ij}]$ be the determinant polynomial.

\begin{definition}  Let $f(x) \in k[x_1, \ldots, x_n]$ be a polynomial.   A {\it determinantal expression of size $m$ for $f$} is an affine linear map $L \co k^n \to k^{m \times m}$ such that $f(x) = \det_m(L(x))$. 
The {\it determinantal complexity of $f$}, denoted by $\dc(f)$, is the smallest $m$ such that there exists a determinantal expression of size $m$ for $f$.
\end{definition}

The main result of this paper is the following  lower bound for the determinantal complexity of a homogeneous polynomial $f \in k[x_1, \ldots, x_n]$.  We denote by $\Sing(f)$ the singular locus of the hypersurface $\V(f) \subseteq k^n$.

\begin{theorem} \label{T:main}
Let $k$ be a field.  Let $f \in k[x_1, \ldots, x_n]$ be a homogeneous polynomial of degree $d > 2$.  If $\codim(\Sing(f)) > 4$, then
$$\dc(f) \ge \codim(\Sing(f)) + 1.$$
\end{theorem}

Our first application of this result is to give a new lower bound for the determinantal complexity of the permanent polynomial
$$\perm_n = \sum_{\sigma \in S_n} x_{1 \sigma(1)} x_{2 \sigma(2)} \cdots x_{n \sigma(n)}.$$
In \cite{valiant1} and \cite{valiant2}, Valiant conjectured that $\dc(\perm_n)$ is not bounded above by any polynomial in $n$ as long as ${\rm char}(k) \neq 2$.  Moreover, Valiant proposed an algebraic analogue of the P versus NP question by introducing the complexity classes $\VP_{\rm e}$ and VNP consisting of sequences of polynomials which are ``computable by arithmetic trees in polynomial time" and ``definable in polynomial time," respectively.
Valiant showed that if $\dc(\perm_n)$ grows faster than any polynomial, then in fact $\VP_{\rm e} \neq \VNP$.

In characteristic $0$, the best known bounds for the determinantal complexity of $\perm_n$ for $n >2$ are: 
\begin{equation} \label{E:perm}
n^2/2 \le \dc(\perm_n) \le 2^n - 1
\end{equation}
where the lower bound was established by Mignon and Ressayre  \cite{mignon-ressayre} (whose argument was subsequently generalized  in \cite{ccl} to characteristic $p \neq 2$ to provide the bound $(n-2)(n-3)/2 \le \dc(\perm_n)$) 
and the upper bound was established by Grenet \cite{grenet} (which holds in any characteristic). 

Since it is known that $\codim(\Sing(\perm_n)) > 4$ for $n>2$ (c.f. \cite[Lem.~2.3]{von-zur-gathen}) if ${\rm char}(k) \neq 2$, we obtain:

\begin{corollary} \label{C:perm}
Let $k$ be a field with ${\rm char}(k) \neq 2$.  If $n > 2$, then 
$$\dc(\perm_n) \ge \codim(\Sing(\perm_n)) + 1.$$
\end{corollary}

For $n=3$, the inequalities \eqref{E:perm} imply that $5 \le \dc(\perm_3) \le 7$.  It has been an open question to determine the actual value of $\dc(\perm_3)$.   Similarly, for $n=4$, the  inequalities \eqref{E:perm} imply that $8 \le \dc(\perm_4) \le 15$.  Since it can be readily computed that $\codim(\Sing(\perm_3)) = 6$ and  $\codim(\Sing(\perm_4)) = 8$ (c.f. \cite{von-zur-gathen}), we obtain:

\begin{corollary} \label{C:perm3}
Let $k$ be a field with ${\rm char}(k) \neq 2$. Then $\dc(\perm_3) = 7$ and $\dc(\perm_4) \ge 9$.
\end{corollary}

It was recently shown in \cite{hi} that if ${\rm char}(k) = 0$, the smallest determinantal expression for $\perm_3$ such that every entry is $0$, $1$ or a variable has size $7$.

\begin{remark}
Since any matrix where the first two columns are zero is a singular point of $\V(\perm_n)$, it follows that $\codim(\Sing(\perm_n)) \le 2n$.
 Therefore, even if the codimension of $\Sing(\perm_n)$ achieves the maximum value $2n$, Corollary \ref{C:perm} would only give the linear bound $\dc(\perm_n) \ge 2n + 1$.  
\end{remark}

Our second application is toward the determinantal complexity of homogeneous forms of low degree.  It is a classical problem to determine which homogeneous forms of degree $d$ have a determinantal expression of size $d$; see \cite{dickson} (or \cite{beauville} for a modern treatment).  
For this discussion, we assume that $k$ is algebraically closed.  
As any binary form $f(x,y)$ of degree $d$ factors into a product of $d$ linear forms, it is clear that $\dc(f) = d$.  It is also known that any (possibly singular) plane curve $f(x,y,z)$ of degree $d$ has $\dc(f) = d$ \cite[Rmk.~4.4]{beauville} and any quadratic form $f$ in $n \le 4$ variables has $\dc(f) = 2$.   It is a classical fact that a general cubic surface ($n=4,d=3$) admits a determinantal expression of size $3$ \cite{schroter}, \cite{cremona}. However, in all other cases ($n=4, d> 3$ or $n > 4$), a dimension count yields that a general homogeneous form $f(x_1, \ldots, x_n)$ of degree $d$ has $\dc(f) > d$ \cite[Thm.~1]{dickson}.  As a direct consequence of Theorem \ref{T:main}, we obtain the following bound for the determinantal complexity for nonsingular hypersurfaces in the case that $n > 4$ and $d \le n$:

\begin{corollary} \label{C:nonsingular}
Let $k$ be a field.   Assume that $n > 4$ and $d \le n$.  If $f(x_1, \ldots, x_n)$ is a non-degenerate homogeneous form (i.e. the hypersurface $\bV(f) \subseteq \PP^{n-1}$ is nonsingular) of degree $d>2$, then
$\dc(f) \ge n + 1$.  
\end{corollary}

In particular, Corollary \ref{C:nonsingular} implies that for $n >4$ and $d \le n$, there does not exist a nonsingular hypersurface of degree $d$ in $\PP^{n-1}$ having a determinantal expression of size $d$.

Using Bertini's theorem, we can obtain the even stronger result:

\begin{theorem} \label{T:bertini}
Let $k$ be an algebraically closed field with ${\rm char}(k) = 0$. Let $L\co k^n\rightarrow k^{m\times m}$ be a linear map with $m\geq 2$. Let $f(x):=det_m(L(x))$.  Then $\codim(\Sing(f)) \le \min(4,n)$ and equality holds if $L$ is a general linear map.  

In particular, if we let $\bV(f) \subseteq \PP^{n-1}$ be the projective hypersurface defined by $f$, then the following statements hold:
\begin{enumerate}
\item if $n\leq 4$ then $\bV(f)$ is nonsingular for a general linear map $L$; and 
\item if $n>4$ then $\bV(f)$ is singular for every linear map $L$.
\end{enumerate}
\end{theorem}

We now consider how the vector space $\Sym^d (k^n)^{\vee}$ of homogeneous forms in $n$ variables of degree $d$ decomposes by the determinantal complexity.   For this discussion, we assume that $k$ is algebraically closed of characteristic $0$.   For  $n=2$ or $3$, this decomposition of $\Sym^d (k^n)^{\vee}$ is trivial as all non-zero forms have determinantal complexity $d$.  The first interesting case is quadratic forms.  For a quadratic form $f$ in $n > 4$ variables of rank $r$, the determinantal complexity of $f$ is $\lceil (r+1)/2 \rceil$ for $r \ge 4$ and $2$ otherwise \cite[Thm.~1.4]{mignon-ressayre}.  Thus, the decomposition of $\Sym^2 (k^2)^{\vee}$ by the determinantal complexity is the same as the stratification by the above function of the rank $r$.

The next interesting case to consider is cubic surfaces.  It is a classical fact that any nonsingular cubic surface has a determinantal expression of size $3$; see \cite{grassman} and \cite{beauville}.  More generally, it was shown in \cite{brundu-logar} that any cubic surface not projectively equivalent to $f = xy^2+yt^2+z^3$ has determinantal complexity $3$ and moreover that $\dc(f) > 3$.  
This form $f$ is also the unique cubic surface (up to projective equivalence) in $\PP^3$ containing a single line.  While it is possible to write down a determinantal expression of size $5$, it has been an open question to determine whether $\dc(f)$ is $4$ or $5$.
Using a similar idea to the proof of Theorem \ref{T:main}, we establish:

\begin{theorem} \label{T:cubic} Let $k$ be a field with ${\rm char}(k) = 0$. Then $\dc(xy^2+yt^2+z^3) = 5$.
\end{theorem}

\begin{remark}
It is worthwhile to include the following well-known observation:  since $\dc(xy^2+yt^2+z^3) > 3$, one sees that the determinantal complexity function $f \mapsto \dc(f)$ is not in general upper semicontinuous (i.e. the locus of forms $f$ with $\dc(f) \ge m$ for a fixed $m$ is not necessarily Zariski-closed).  Indeed,  the cubic surface $xy^2+yt^2+z^3$ degenerates to singular cubic surfaces (e.g.  $z^3 = \lim_{\epsilon \to 0} \epsilon xy^2+ \epsilon yt^2+z^3$) with determinantal complexity $3$.  Nevertheless, for each $m$, the locus of forms $f$ with $\dc(f)=m$ is constructible.
\end{remark}

\subsection*{Acknowledgements}  We thank Nicolas Ressayre for raising the question regarding the value of $\dc(xy^2+yt^2+z^3) $ during the problem discussion at the {\it Geometric Complexity Theory} workshop at the Simons Institute in Berkeley in September, 2014.

\section{Proofs}

We begin with an easy generalization of \cite[Thm.~3.1]{von-zur-gathen}.

\begin{proposition}  \label{P:rank} 
Let $k$ be any field.  Let $f \in k[x_1, \ldots, x_n]$ be a homogeneous polynomial $f \in k[x_1, \ldots, x_n]$ satisfying $\codim(\Sing(f)) > 4$.  If $L \co k^n \to k^{m \times m}$ is a determinantal expression for $f$, then $\im(L) \cap \Sing(\det_m) = \emptyset$.  In particular, $\rank(L(0)) = m-1$.
\end{proposition}

\begin{proof}  Let $y_{kl}$ for $1 \le k,l \le m$ be coordinates on $k^{m \times m}$ and write $L = (L_{kl})$ where each $L_{kl}$ is an affine linear form in $x_1, \ldots, x_n$.  By the chain rule, 
$$\frac{\partial f}{\partial x_i} (x) = \sum_{1 \le k,l \le m} \frac{\partial \det_m}{\partial y_{kl}} (L(x))  \frac{\partial L}{\partial x_i } (x)$$
and it follows that $L^{-1}(\Sing(\det_m)) \subseteq \Sing(f)$.  If $L^{-1} (\Sing(\det_m)) \neq \emptyset$, then 
$$\codim(\Sing(f)) \le \codim(L^{-1}(\Sing(\det_m))) \le \codim(\Sing(\det_m))=4$$
where the second  inequality is the standard bound for the codimension of an inverse image (c.f. \cite[Thm.~17.24]{harris}) and the last equality follows from the fact that the singular locus of $\det_m$ consists of matrices $A$ with $\rank(A) \le m-2$.  The above inequalities contradict our hypothesis that $\codim(\Sing(f)) > 4$.  

For the final statement, we know that $\rank(L(0)) \ge m-1$.  But since $f$ is homogeneous, we have $\det(L(0)) = f(0) = 0$ which implies that $\rank(L(0)) = m-1$.
\end{proof}

\begin{remark} Proposition \ref{P:rank} is true more generally (with the same proof) for any morphism $L \co k^n \to k^{m \times m}$ of varieties such that $f(x) = \det_m(L(x))$.
\end{remark}

\begin{proof}[Proof of Theorem \ref{T:main}] Let $L \co k^{n}\rightarrow k^{m\times m}$ be a determinantal expression for $f$.   Proposition \ref{P:rank} implies that $\rank (L(0)) = n-1$.
By multiplying $L$ by matrices on the left and right, we may assume that $J= L(0)$ is the $m \times m$ matrix $(J_{ij})$ with $J_{ii}=1$ for $2\leq i \leq m$ and $J_{ij}=0$ otherwise.  Therefore, $L(x) = J + Z(x)$ where $Z=(Z_{ij})$ is an $m\times m$ matrix where each $Z_{ij}$ is a linear form in $x=(x_{1}, \ldots, x_n)$. Since $f(x) = \det_m(J+Z(x))$ and the left hand side is homogeneous of degree $d > 2$, we conclude that the equations $Z_{11}=0$ and 
$\sum_{j=2}^m Z_{1j}Z_{j1}=0$ hold.

Let $I\subseteq k[x_1, \ldots, x_n]$ be the ideal generated by the first row and column of $Z$. Since $Z_{11}=0$, every summand of $\det(J+Z(x))$ is divisible by a product of two elements of $I$. In particular, $f \in I^2$ and all partial derivatives $\partial f/\partial x_i$ are in  $I$.   Thus, $\V(I)\subseteq \Sing(f)$. 

We now obtain an upper bound on the codimension of $\V(I)$ as follows.  We introduce the linear map
$$\begin{aligned}
G \co k^{n} & \to  k^{2(m-1)} \\
(x_1, \ldots, x_n) & \mapsto \left(Z_{12},Z_{13},\ldots, Z_{1m}, Z_{21}, Z_{31}, \ldots, Z_{m1}\right).
\end{aligned}$$
Since $\V(I)=\ker(G)$, we have that $\codim(\V(I)) = \dim(\im(G))$.  Now  $\im(G)$ is a linear subspace which is entirely contained in the non-degenerate quadric $\sum_{j=2}^m w_{1j}w_{j1}=0$ in $k^{2(n-1)}$ where $w_{12}, \ldots, w_{1m}, w_{21}, \ldots, w_{m1}$ are the coordinates on $k^{2(n-1)}$. Such subspaces have dimension at most $m-1$ (c.f.~\cite[pg.~289]{harris})\footnote{If ${\rm char}(k) = 2$, one checks that the argument of \cite[pg.~289]{harris} applies to  $Q(w)=\sum_{j=2}^m w_{1j}w_{j1}$ using the bilinear form $Q_0(w, w') = Q(w+w') - Q(w) - Q(w')$.}.  We conclude that 
$$\codim(\Sing(f)) \le \codim(\V(I)) \le m-1$$
so that $m \ge \codim(\Sing(f)) + 1$ as claimed.
\end{proof}

\begin{proof}[Proof of Theorem \ref{T:bertini}]
We begin by making the following observation.  If $L\co k^n\rightarrow k^{m\times m}$ is not injective then $\V(f) \subseteq k^n$ is a cone over the hypersurface $\V(f') \subseteq k^s$ where $f'$ is the determinant of the linear map $k^{s} = k^n / \ker(L) \to k^{m \times m}$ (i.e. there is a choice of basis $x_1, \ldots, x_s, x_{s+1}, \ldots, x_n$ for $\PP^{n-1}$ such that the hypersurface defined by $L$ does not involve the variables $x_{s+1}, \ldots, x_{n}$).  It follows that the codimension of $\Sing(f)$ in $k^n$ is equal to the codimension of $\Sing(f')$ in $k^{s}$.  

We first show that $\codim(\Sing(f)) \le 4$.  By the above observation, we may assume that $L$ is injective and that $n\leq m^2$.  Therefore, $\V(f)$ is the intersection of the determinantal hypersurface $\V(\det_m) \subseteq k^{m \times m}$ with $m^2-n$ hyperplanes. 
  As the intersection of a variety with a singular point $p$ with a hypersurface passing through $p$ also has a singularity at $p$, the codimension of the singular locus can only decrease after intersecting.  As $\codim( \Sing(\det_m)) = 4$, we see that $\codim( \Sing(f)) \le 4$.

  We now show that $\codim(\Sing(f)) = \min(4,n)$ for general linear maps $L$.  By the above observation, we may assume that $L$ is injective and $n \le m^2$. Then the projective hypersurface $\bV(f) \subseteq \PP^{n-1}$ is the intersection of $\bV(\det_m)\subseteq \PP^{m^2 -1}$ with $m^2-n$ general hyperplanes.  
We recall Bertini's Theorem (c.f. \cite[Thm.~17.16]{harris}):
if $X\subseteq \PP^N$ is any variety over $k$ and $H$ is a general hyperplane, then
 $\Sing( H \cap X) = H \cap \Sing(X)$, and moreover $\dim(\Sing(H \cap X)) = \dim(\Sing(X)) - 1$ provided that $\dim(\Sing(X)) > 0$.

We apply Bertini's Theorem $m^2-n$ times. Since the singular locus of $\bV(\det_m)$ in $\PP^{m^2-1}$ has dimension $m^2-5$,  we can conclude that the singular locus of a hypersurface $\bV(f) \subseteq \PP^{n-1}$ defined by a general linear map $L$ has dimension $m^2-5-(m^2-n)=n-5$ if $n \ge 5$ and is empty if $n < 5$, which was our intended goal.

For the final statements, if $n \le 4$ (resp. $n > 4$), then we have shown that $\codim (\Sing(f)) =  n$ for general linear maps $L$ (resp. $\codim( \Sing(f)) \le 4$ for every linear map $L$) which implies that $\bV(f) \subseteq \PP^{n-1}$ is nonsingular for general $L$ (resp. singular for every $L$).
\end{proof}

\begin{proof}[Proof of Theorem \ref{T:cubic}]
By \cite[Prop.~4.3]{brundu-logar}, we know that $\dc(f) > 3$.  If $\dc(f) = 4$, then we can assume that there is a determinantal expression  $L = J +Z \co k^4 \to k^{4 \times 4}$, where $Z=(Z_{ij})$ is a matrix of linear forms and $J = (J_{ij})$ is the matrix of rank $r$ with $J_{ii} = 1$ for $5-r \le i \le 4$ and $0$ otherwise.  We will show that each possibility for the rank $r$ yields a contradiction.  First, the rank $r$ cannot be $0$ or $4$ as $f$ is homogeneous of degree $3$.   If $r=1$, the degree $3$ component of $\det(L)$, namely $\det (Z_{ij})_{1 \le i,j, \le 3}$, gives a determinantal expression of $f$ of size $3$, contradicting the fact that $\dc(f) > 3$. 

If $r=2$, then $Z_{11}Z_{22} - Z_{12}Z_{21} = 0$.  We will argue that we can reduce to the case that $Z_{11}=Z_{21} = 0$. Indeed, if $Z_{11}=0$, then either $Z_{21}=0$ or $Z_{12}=0$, and in the latter case we replace $L$ with its transpose.  If $Z_{11} \neq 0$, then either $Z_{12}$ or $Z_{21}$ is a multiple of $Z_{11}$, and after potentially replacing $L$ with its transpose, we may assume that $Z_{12} \in \langle Z_{11} \rangle$.  We may replace $L$ by $P^{-1} L P$ where $P$ is an invertible matrix of constants, so that $Z_{12}=0$.  But then $Z_{22}=0$ and the claim is established by interchanging the first and second column (and negating one column).  Since $y=z=0$ is the unique line contained in this cubic surface, we can replace $L$ with $PLP^{-1}$ so that $Z_{31}=y$ and $Z_{41}=z$.  This yields
$$\begin{aligned} f &= \det 
\begin{pmatrix} 
	0 		& Z_{12} & Z_{13} 		& Z_{14} \\
	0	 	& Z_{22} & Z_{23} 		& Z_{24} \\
	y		& Z_{32} & 1 + Z_{33}	& Z_{34} \\
	z		& Z_{42} & Z_{43}		& 1 + Z_{44} 
\end{pmatrix} \\
	& = y(Z_{12}Z_{23}  - Z_{13} Z_{22}  ) - z ( Z_{14}Z_{22}  - Z_{12} Z_{24}).
\end{aligned}
$$
Since $f = 0$ along the subspace $Z_{12} = Z_{22} =0$, we can use the fact again that $y=z=0$ is the unique line in this cubic surface to replace $L$ with $P^{-1} L P$ so that $Z_{12} = y$ and $Z_{22} = z$.  But then $f \in (y,z)^2$, a contradiction.

Finally, suppose $r=3$.  By the argument in the proof of Theorem \ref{T:main}, we know that $Z_{11} = 0$ and $Z_{12} Z_{21} + Z_{13} Z_{31} + Z_{14} Z_{41} = 0$.  Moreover, we know that the dimension of the subspace $I_1 := \langle Z_{12},Z_{13},Z_{14}, Z_{21}, Z_{31}, Z_{41} \rangle$ is at most $3$ and that if $I$ denotes the ideal generated by $I_1$, then $\Jac(f) := (f_x,f_y,f_z,f_t) \subseteq I$.  The equation $\V(f) \subseteq k^4$ is singular along $y=z=t=0$ which yields that $(y,z,t) = \sqrt{\Jac(f)} \subseteq I$ and that $I_1 = \langle y, z, t \rangle$.   Either the span of the first row or the first column must be equal to $I_1$; otherwise, as $y=z=0$ is the unique line in the cubic surface, both spans would be equal to $\langle y,z \rangle$ contradicting that $I_1 = \langle y, z, t \rangle$.  Therefore, after replacing $L$  by $P^{-1} L P$ or $P^{-1} L^{\intercal} P$, we can assume that $Z_{21} = z$, $Z_{31}=y$ and $Z_{41}=t$.  As the matrix expressing $Z_{12}, Z_{13}, Z_{14}$ in terms of $z,y,t$ is necessarily anti-symmetric, we can write $Z_{12} = \alpha t + \beta y$, $Z_{13} = -\beta z + \gamma t$, and $Z_{14} = -\gamma y -\alpha z$.  To summarize, we have
$$\begin{aligned} f &= \det 
\begin{pmatrix} 
	0 		&  \alpha t + \beta y & -\beta z + \gamma t 		& -\gamma y -\alpha z \\
	z	 	& 1+Z_{22} & Z_{23} 		& Z_{24} \\
	y		& Z_{32} & 1 + Z_{33}	& Z_{34} \\
	t		& Z_{42} & Z_{43}		& 1 + Z_{44} 
\end{pmatrix} .
\end{aligned}
$$
Comparing the coefficients in the above expression 
of the 6 monomials of degree $3$ whose $x$-exponent is $1$, 
one obtains six equations that the coefficients $\mathrm{X}_{ij}$ of $x$ in $Z_{ij}$ must satisfy:
$$\begin{array} {l r l}
xy^2 \co &\beta \rX_{23} - \gamma \rX_{43}  = 1 \\
xz^2 \co &- \beta \rX_{32} - \alpha \rX_{42}  = 0 \\
xt^2 \co &\alpha \rX_{24} + \gamma \rX_{34}  = 0
\end{array}
\qquad \quad
\begin{array} {l r  l}
xyz \co &\beta (\rX_{22} - \rX_{33}) - \gamma \rX_{42} - \alpha \rX_{43}   = 0 \\
xyt \co &\gamma (\rX_{33} - \rX_{44}) + \alpha \rX_{23} + \beta \rX_{24} = 0 \\
xzt \co & \alpha(\rX_{22} - \rX_{44}) + \gamma \rX_{32} - \beta \rX_{34} = 0
\end{array}$$
  One can check that these equations are inconsistent unless $\alpha = 0$ and $ \gamma \neq 0$.  In this latter case, the top row of $L$ is 0 if $y = -\beta z + \gamma t = 0$ which in turn implies that $f$ vanishes on this subspace, a contradiction.  We have therefore established that  $\dc(f) > 4$. On the other hand, one can check that
$$f = \det \left(\begin{array}{rrrrr}
- y & z & 0 & 0 & 0 \\
0 & 0 & z & t & x \\
z & 0 & 1 & 0 & 0 \\
0 & t & 0 & 1 & 0 \\
0 & y & 0 & 0 & 1
\end{array}\right)$$
which implies that $\dc(f) = 5$.
\end{proof}

\bibliography{refs}
\bibliographystyle{amsalpha}
\end{document}